\documentclass[a4paper,11pt]{article}
\usepackage[margin=3cm]{geometry}

\usepackage{graphicx}
\usepackage{xcolor}
\usepackage{hyperref}
\usepackage{amssymb,amsmath,amsthm}
\usepackage[affil-it]{authblk}
\graphicspath{{img/}}

\theoremstyle{plain}
\newtheorem{theorem}{Theorem} 
\newtheorem{lemma}{Lemma}

\usepackage{csquotes}
\usepackage{subcaption}
\usepackage{booktabs}
\usepackage[inline]{enumitem}
\setlist[enumerate]{nosep}
\usepackage{wrapfig}

\DeclareMathOperator{\outl}{out}
\DeclareMathOperator{\inl}{in}
\DeclareMathOperator{\lca}{lca}
\DeclareMathOperator{\dom}{dom}
\DeclareMathOperator{\slope}{slope}
\DeclareMathOperator{\parent}{par}

\makeatletter
\renewcommand\paragraph{\@startsection{paragraph}{4}{\z@}%
                                    {\topsep}%
                                    {-1em}%
                                    {\normalsize\bfseries}}
\renewcommand\subparagraph{\@startsection{subparagraph}{5}{\z@}%
                                    {\topsep}%
                                    {-1em}%
                                    {\normalfont\normalsize\itshape}}
\makeatother

\title{Drawing Planar Graphs with Few Geometric Primitives\thanks{%
A preliminary version of this work appeared in: Proc. 43rd Int. Workshop Graph-Theor. Concepts Comput. Sci. (WG'17)~\cite{hkms-dpgfg-WG17}.
The work of P.~Kindermann and A.~Schulz was supported by DFG grant SCHU~2458/4-1.
  The work of W.~Meulemans was supported by Marie Sk\l{}odowska-Curie Action MSCA-H2020-IF-2014 656741.}}

\author[1]{Gregor~H\"ultenschmidt}
\author[2]{Philipp~Kindermann}
\author[3]{Wouter~Meulemans}
\author[1]{Andr\'e~Schulz}
\affil[1]{LG Theoretische Informatik, FernUniversit\"at in Hagen, Germany\\
\texttt{\href{mailto:gregorhueltenschmidt@gmx.de}{gregorhueltenschmidt@gmx.de}\\ \href{mailto:andre.schulz@fernuni-hagen.de}{andre.schulz@fernuni-hagen.de}}}
\affil[2]{David R. Cheriton School of Computer Science, University of Waterloo, Canada\\
\texttt{\href{mailto:pkinderm@uwaterloo.ca}{pkinderm@uwaterloo.ca}}}
\affil[3]{TU Eindhoven, The Netherlands\\
\texttt{\href{mailto:w.meulemans@tue.nl}{w.meulemans@tue.nl}}}

\begin{document}

\maketitle

\begin{abstract}
We define the \emph{visual complexity} of a plane graph drawing to be the number of basic geometric objects needed to represent all its edges.
In particular, one object may represent multiple edges (e.g., one needs only one line segment
to draw a path with an arbitrary number of edges). Let~$n$ denote the number of vertices of a graph.
We show that trees can be drawn with~$3n/4$ straight-line segments on a polynomial grid, and with~$n/2$ straight-line segments on a quasi-polynomial grid. Further, we present an algorithm for drawing planar 3-trees with $(8n-17)/3$ segments on
an~$O(n)\times O(n^2)$ grid. This algorithm can also be used with a small modification to draw maximal outerplanar
graphs with $3n/2$ edges on an~$O(n)\times O(n^2)$ grid.
We also study the problem of drawing maximal planar graphs with circular arcs and provide an algorithm to draw such graphs using only $(5n - 11)/3$ arcs.
This is significantly smaller than the lower bound of $2n$ for line segments for a nontrivial graph class.
\end{abstract}

\section{Introduction}

The complexity of a graph drawing can be assessed in a variety of ways: area, crossing number, bends, angular resolution, etc.
All these measures have their justification, but in general it is impossible to optimize all of them in a single drawing. More
recently, the \emph{visual complexity} was suggested as another quality measure for drawings~\cite{schulz2015}. The visual complexity denotes the number of simple geometric entities used in the drawing.

\begin{table}[t]
\centering
\caption{Upper bounds on the visual complexity. Here, $n$ is the number of vertices, $\vartheta$ the number of odd-degree vertices and $e$ the number of edges. Constant additions or subtractions have been omitted.}
\label{tbl:results}
\medskip
\begin{tabular}{l r c r  c}
\toprule
\textbf{Class} & \multicolumn{2}{c}{\textbf{Segments}\qquad\qquad\qquad} & \multicolumn{2}{c}{\textbf{Arcs}}\\
\midrule
trees & $\vartheta/2$ & \cite{dujmovic2007} &$\vartheta/2$  & \cite{dujmovic2007} \\
maximal outerplanar & $n$ & \cite{dujmovic2007} & $n$  & \cite{dujmovic2007} \\
3-trees & $2n$ & \cite{dujmovic2007} & $11e/18$ &\cite{schulz2015} \\
3-connected planar & $5n/2$ & \cite{dujmovic2007} & $2e/3$ & \cite{schulz2015} \\
cubic 3-connected planar & $n/2$ & \cite{igamberdiev2015,mondal2013} &  $n/2$ & \cite{igamberdiev2015,mondal2013}  \\
triangulation & $7n/3$ & \cite{durocher2014} & $\bf 5n/3$ & \bf Thm.~\ref{thm:triangle-arc} \\
4-connected triangulation & $9n/4$ & \cite{durocher2014} & $\bf 3n/2$ & \bf Thm.~\ref{thm:4-conn-triangle-arc} \\
4-connected planar & $9n/4$ & \cite{durocher2014} & $\bf 9n/2 - e$ & \bf Thm.~\ref{thm:4-conn-planar-arc} \\
planar & $16n/3 - e$ & \cite{durocher2014} & $\bf 14n/3 - e$ & \bf Thm.~\ref{thm:planar-arc} \\
\bottomrule
\end{tabular}
\end{table}

Typically, we consider as entities either straight-line segments or circular arcs.
To distinguish these two types of drawings, we call the former
\emph{segment drawings} and the latter \emph{arc drawings}.
The idea is that we can use, for example, a single segment to draw a path of collinear edges. The hope
is that a drawing that consists of only a few geometric entities is easy to perceive.
A recent user study~\cite{kms-eaadf-gd17} suggests that visual complexity may positively influence
aesthetics, depending on the background of the observer, as long as it does not
introduce unnecessarily sharp corners.
It is natural to ask for a drawing of a graph with the smallest visual complexity.
Unfortunately, it is an ${\sf NP}$-hard problem to determine the smallest number of segments necessary in a segment drawing~\cite{durocher2013}.
However, we can still expect to find algorithms, which guarantee bounds for certain graph classes.

\paragraph{Related work.}
For a number of graph classes, upper and lower bounds are known for segment
drawings and arc drawings; the upper bounds are summarized in
Table~\ref{tbl:results}. However, these bounds (except for cubic 3-connected
graphs) do not require the drawings to be on the grid.
In his thesis, Mondal~\cite{mondal-thesis} gives an algorithm for triangulations
with $n$ vertices using
$8n/3-O(1)$ segments on a grid of size $2^{O(n\log n )}$ in general and
of size $2^{O(n)}$ for triangulations of bounded degree. Even with this large
grid, the algorithm uses
substantially more segments than the best-known algorithm for triangulations
without the grid requirement by
Durocher and Mondal~\cite{durocher2014}, which uses $7n/3-O(1)$ segments.

There are three trivial lower bounds for the number of segments required to
draw any graph~$G=(V,E)$ with~$n$ vertices and~$e$ edges:
\begin{enumerate}[label=(\roman*)]
	\item $\vartheta/2$, where $\vartheta$ is the number of odd-degree vertices,
  \item $\max_{v\in V} \lceil \deg(v)/2\rceil$, and
  \item $\lceil e/(n-1)\rceil$.
\end{enumerate}
For triangulations and general planar graphs, a lower bound of $2n - 2$ and
$5n/2 - 4$, respectively, is known \cite{dujmovic2007}.
Note that the trivial lower bounds are the same as for the slope number
of graphs~\cite{wc-dcgms-CJ94}, that is, the minimum number of slopes required
to draw all edges, and that the number of slopes of a drawing is upper
bounded by the number of segments.
Chaplick et al.~\cite{cflrvw-dgflf-GD16,cflrvw-cdgfl-WADS17} consider a
similar problem where all edges are to be covered by few lines (or planes);
the difference to our problem is that collinear segments are counted only once in their model.
In the same fashion, Kryven et al.~\cite{krw-dgfcf-CALDAM18} aim to cover all
edges by few circles (or spheres).

\paragraph{Contributions.} In this work, we present two types of results. In the first part (Sections~\ref{sect:trees},~\ref{sect:3trees} and~\ref{sect:outerplanar}),
we present algorithms for segment drawings on the grid with low visual complexity.
This direction of research was posed as an open problem
by~Dujmovi{\'c} et al.~\cite{dujmovic2007}, but only a few results exist; see Table~\ref{tab:gridresults}.
We present an algorithm that draws trees on an $O(n^2)\times O(n^{1.58})$ grid using $3n/4$ straight-line segments. For comparison, the drawings of Schulz~\cite{schulz2015} need also~$3n/4$ arcs on a smaller $O(n)\times O(n^{1.58})$ grid, but use the more complex circular arcs instead.  Our segment drawing algorithm for trees can be modified to generate drawings with an optimal number of $\vartheta/2$ segments on a quasi-polynomial grid.
We also present algorithms to compute segment drawings of planar 3-trees and maximal outerplanar
graphs, both on an~$O(n)\times O(n^2)$ grid. In the case of planar 3-trees, the algorithm
needs at most~$(8n-17)/3$ segments, and in the case of maximum outerplanar graphs the algorithm
needs at most~$3n/2$ segments.

Finally, in sections~\ref{sect:triangulations} and~\ref{sect:grapharcs}, we study arc drawings of triangulations and general planar graphs.
In particular, we prove that $(5n-11)/3$ arcs are sufficient to draw any triangulation with $n$ vertices.
We highlight that this bound is significantly smaller than the $2n - 2$ \emph{lower} bound known for segment drawings~\cite{dujmovic2007} and the so far best-known $2e/3 = 2n$ upper bound for circular arc drawings~\cite{schulz2015}.
A straightforward extension shows that $(14n - 3e - 29)/3$ arcs are sufficient for general planar graphs with $e$ edges.

\begin{table}[t]
	\centering
\caption{Same as in Table~\ref{tbl:results} but for grid drawings.}
\label{tab:gridresults}
\begin{tabular}{ll@{\quad}l@{\quad}l@{\quad}l}
\toprule
\bf Class & \bf Type & \bf Compl. & \bf Grid & \bf Ref. \\
\midrule
trees & arcs &  $3n/4$& $O(n)\times O(n^{1.58})$ & \cite{schulz2015} \\
trees & segments &  $\bf 3n/4$& $\bf O(n^2)\times O(n^{1.58})$ &
\bf Thm.~\ref{thm:treepoly} \\
trees & segments &  $\bf \vartheta/2$& \bf quasi-polynom. &
\bf Thm.~\ref{thm:treequasipoly} \\
cubic planar 3-conn. & segments &  $n/2$ & $O(n)\times O(n)$ &
\cite{igamberdiev2015,mondal2013} \\
maximal outerplanar & segments &  $\bf 3n/2$ & $\bf O(n)\times O(n^2)$ &
\bf Thm.~\ref{thm:outerplanar} \\
triangulation & segments & $ 8n/3$ & $O((3.63n)^{4n/3})$ &
\cite{mondal-thesis} \\
planar 3-tree & segments & $\bf 8n/3$ & $\bf O(n)\times O(n^2)$ &
\bf Thm.~\ref{thm:planar3tree} \\
\bottomrule
\end{tabular}
\end{table}

\paragraph{Preliminaries.}
Given a triangulated planar graph~$G=(V,E)$ on~$n$ vertices, a \emph{canonical
order} $\sigma=(v_1,\ldots,v_n)$ is an ordering of the vertices in~$V$ such
that, for~$3\le k\le n$, (i)~the subgraph~$G_k$ of~$G$ induced
by~$v_1,\ldots,v_k$ is biconnected, (ii)~the outer face of~$G_k$ consists of
the edge~$(v_2,v_1)$ and a path~$C_k$, called \emph{contour}, that contains~$v_k$,
and (iii)~the neighbors of~$v_k$ in~$G_{k-1}$ form a subpath
of~$C_{k-1}$~\cite{dfpp-sssfe-STOC88,dpp-hdpgg-C90}. A canonical order can be
constructed in reverse order by repeatedly removing a vertex without an incident
chord from the outer face.

Most of our algorithms make ample use of \emph{Schnyder realizers}~\cite{schnyder1990}. Assume we selected a face as the outer
face with vertices $v_1$, $v_2$ and $v_n$.
We decompose the interior edges into three trees: $T_1$, $T_2$, and $T_n$ rooted
at $v_1$, $v_2$, and $v_n$, respectively. The edges of the trees are oriented to their roots. For $k\in\{1,2,n\}$, we call each edge in~$T_k$ a \emph{$k$-edge}
and the parent of a vertex in~$T_k$ its \emph{$k$-parent}.
In the figures of this paper, we will draw 1-edges red, 2-edges blue, and
$n$-edges green. The decomposition is a Schnyder realizer if at every interior vertex the edges are cyclically ordered as: outgoing 1-edge, incoming $n$-edges, outgoing 2-edge, incoming 1-edges, outgoing $n$-edge, incoming 2-edges. The trees of a Schnyder realizer are also called
\emph{canonical ordering trees}, as each describes a canonical order on the
vertices of~$G$ by a (counter-)clockwise pre-order traversal~\cite{fraysseix2001}.
There is a unique
\emph{minimal} realizer such that any interior cycle in the union of the three
trees is oriented clockwise \cite{brehm2000}; this realizer can be computed in
linear time \cite{brehm2000,schnyder1990}. The number of such cycles is denoted
by $\Delta_0$ and is upper bounded by $\lfloor(n-1)/2\rfloor$~\cite{zhang2005}.
Bonichon~et~al.~\cite{bonichon2002} prove that the total number of leaves in a
minimal realizer is at most $2n-5-\Delta_0$.

\section{Trees with segments on the grid}\label{sect:trees}

Let~$T=(V,E)$ be an undirected tree. Our algorithm follows the basic idea of the
circular arc drawing algorithm
by Schulz~\cite{schulz2015}. We make use of the \emph{heavy path
decomposition}~\cite{tarjan1983} of trees, which is defined as follows.
First, root the tree at some vertex~$r$ and direct all edges away from the root.
Then, for each non-leaf~$u$, compute the
size of each subtree rooted in one of its children. Let~$v$ be the child of~$u$
with the largest subtree (one of them in case of a tie).
Then,~$(u,v)$ is called a \emph{heavy edge} and
all other outgoing edges of~$u$ are called \emph{light edges}.
The maximal connected
components of heavy edges form the \emph{heavy paths} of the decomposition.

We call the vertex of a heavy path closest to the root its \emph{top node}
and the subtree rooted in the top node a \emph{heavy path subtree}. We define
the \emph{depth} of a heavy path (subtree) as follows. We treat
each leaf that is not incident to a heavy edge as a heavy path of depth~0.
The depth of each other heavy path~$P$ is by~1 larger than the maximum depth of
all heavy paths that are connected from~$P$ by an outgoing light edge.
Heavy path subtrees of common depth are disjoint.

\paragraph{Boxes.}
We order the heavy paths nondecreasingly by their depth and then draw their
subtrees in this order. Each heavy path subtree is placed completely inside an
L-shaped box (\emph{heavy path box}) with its top node placed at the
reflex angle; see Fig.~\ref{fig:lshape-1}
for an illustration of a heavy path box~$B_i$ with top node~$u_i$,
width~$w_i=\ell_i+r_i$, and
height~$h_i=t_i+b_i$.
We require that
\begin{enumerate}[label=(\roman*)]
	\item heavy-path boxes of common depth are disjoint,
	\item $u_i$ is the only vertex on the boundary of $B_i$, and
	\item $b_i\ge t_i$.
\end{enumerate}
Note that the boxes will be mirrored horizontally and/or vertically in
some steps of the algorithm. We assign to each heavy path subtree of depth~0 a
heavy path box~$B_i$ with $\ell_i=r_i=t_i=b_i=1$.

\begin{figure}[t]
	\centering
	\begin{subfigure}[b]{.3\linewidth}
		\centering
		\includegraphics[page=1]{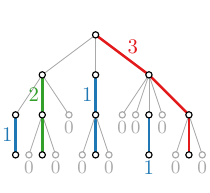}
		\caption{A heavy path decomposition and the depths.}
		\label{fig:heavypath}
	\end{subfigure}
	\hfill
	\begin{subfigure}[b]{.3\linewidth}
		\centering
		\includegraphics[page=1]{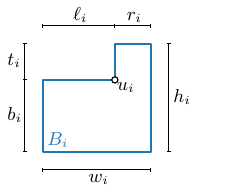}
		\caption{The heavy path box~$B_i$ with top node~$u_i$.}
		\label{fig:lshape-1}
	\end{subfigure}
	\hfill
	\begin{subfigure}[b]{.3\linewidth}
		\centering
		\includegraphics[page=2]{lshape}
		\caption{The merged box~$B^*_i$ for~$B_{2i-1}$ and~$B_{2i}$.}
		\label{fig:lshape-2}
	\end{subfigure}
	\caption{The heavy-path boxes and merged boxes with their respective lengths.}
	\label{fig:lshape}
\end{figure}

\paragraph{Drawing.}
Since a heavy path subtree of depth~$0$ consists of exactly one vertex
that is placed in a box as described above, we have already drawn all
heavy path subtrees of depth~$0$. We will inductively draw the heavy path
subtrees by their depth.
Assume that we have already drawn each heavy path subtree of depth~$d\ge 0$.
In the $(d+1)$-th step of our algorithm, we will draw all heavy path subtrees
of depth~$d+1$. Let $\langle v_1,\ldots,v_m\rangle$ be a heavy path of
depth~$d+1$. We proceed as follows; see Fig.~\ref{fig:treerecursion-1} for an
illustration. The last vertex on a heavy path has to be a leaf, so~$v_m$
is a leaf.

If the outdegree of $v_{m-1}$ is odd, then we place
the vertices $v_1,\ldots,v_m$ on a vertical line; otherwise, we place only
the vertices $v_1,\ldots,v_{m-1}$ on a vertical line and treat~$v_m$ as a
heavy path subtree of depth~0 that is connected to~$v_{m-1}$.

For $1\le h\le m-1$,
all heavy-path boxes adjacent to~$v_h$ will be drawn either in a rectangle
on the left side of the edge $(v_h,v_{h+1})$ or in a rectangle on the right side
of the edge $(v_{h-1},v_h)$ (a rectangle that has~$v_1$ as its bottom left
corner for $h=1$).

We now describe how to place the heavy-path boxes~$B_1,\ldots,B_k$ with top
node~$u_1,\ldots,u_k$, respectively, adjacent to
some vertex~$v$ on a heavy path into the rectangles described above;
see Fig.~\ref{fig:treerecursion-2} for an illustration.

First, assume that~$k$ is even. Then, for $1\le i\le k/2$, we order the boxes
such that $b_{2i}\le b_{2i-1}$. We place the
box~$B_{2i-1}$ in the lower left rectangle and box~$B_{2i}$ in the upper
right rectangle in such a way that the edges $(v,u_{2i-1})$ and $(v,u_{2i})$
can be drawn with a single segment. To this end, we construct a \emph{merged}
box~$B^*_i$ as depicted in Fig.~\ref{fig:lshape-2} with
$\ell_i^*=\max\{\ell_{2i-1},\ell_{2i}\}$, $r_i^*=\max\{r_{2i-1},r_{2i}\}$,
and $w_i^*=\ell_i^*+r_i^*$; the heights are defined analogously.
The merged boxes will help us reduce the number of segments.

We mirror all merged boxes horizontally and place them in the lower left
rectangle (of width $\sum_{i=1}^{k/2}w^*_i$) as follows.
We place~$B_1^*$ in the top left corner of the rectangle.
For~$2\le j\le k/2$, we place~$B_j^*$ directly to the right of~$B_{j-1}^*$
such that its top border lies exactly~$t_{j-1}^*$ rows below the top border
of~$B_{j-1}^*$. Symmetrically, we place the merged boxes (vertically mirrored)
in the upper right rectangle.

Finally, we place each box~$B_{2i-1}$
(horizontally mirrored) in the lower left copy of $B_i^*$ such that their inner
concave angles coincide, and we place each box~$B_{2i}$ (vertically mirrored) in
the upper right copy of $B_i^*$ such that their inner concave angles coincide.

If~$k$ is odd, then we simply add a dummy
box~$B_{k+1}=B_k$ that we remove afterwards.

The box of the heavy path subtree $\langle v_1,\ldots,v_m\rangle$ is the smallest
box that contains $v_1,\ldots,v_m$ as well as the rectangles next to these
vertices (see Fig~\ref{fig:treerecursion-1}). We proceed the same way for every other heavy path subtree of
depth~$d+1$; thus, we have obtained a drawing and its box for every
heavy path subtree of depth~$d+1$.

\begin{figure}[t]
	\centering
	\subcaptionbox{
		\label{fig:treerecursion-1}}
		{\centering
		\includegraphics[page=1]{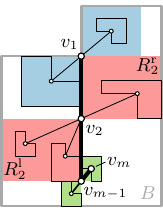}
  }
	\hfil
	\subcaptionbox{
		\label{fig:treerecursion-2}}
		{\centering
		\includegraphics[page=2]{treerecursion}
	}
	\caption{(a) Placement of a heavy path, its box~$B$, and areas for the adjacent heavy path
	  boxes. (b) Placement of the heavy-path boxes adjacent to~$v$.}
	\label{fig:treerecursion}
\end{figure}

\paragraph{Analysis.}
We will now calculate the width and the height of this construction.

\begin{lemma}\label{lem:area}
  Any heavy path subtree of
  depth~$d$ that contains~$n'$ vertices is drawn in a heavy-path box~$B$ with
  width $w=\ell+r\le 2^{d+1}\cdot n'$ and height $h=b+t\le 2\cdot (3/2)^d \cdot n'$
  and~$b\ge t$, where~$\ell,r,b,t$ are defined as in Fig.~\ref{fig:lshape-1}.
\end{lemma}
\begin{proof}
  We prove the lemma by induction over~$d$.

  Recall that we place the depth-0 heavy paths in a box of width $2=2^1\cdot 1$
  and height $2=2\cdot (3/2)^0\cdot 1$, so the bounds hold for~$d=0$.
  Assume that the bound holds for all heavy path subtrees of depth~$d\ge 0$.

  Let~$\langle v_1,\ldots,v_m\rangle$ be a heavy path of depth~$d+1$.
  Consider the drawing of the heavy path subtree rooted in~$v_1$ as described above.
  Let~$R_j^\mathrm{l}$ be the left rectangle
  and let~$R_j^\mathrm{r}$ be the right rectangle that contain all heavy-path
  boxes adjacent to~$v_j$ for~$1\le j\le m-1$; see Fig.~\ref{fig:treerecursion-1}. We denote by~$w_j^\mathrm{l}$
  and~$h_j^\mathrm{l}$ the width and the height of~$R_j^\mathrm{l}$, respectively,
  and by~$w_j^\mathrm{r}$
  and~$h_j^\mathrm{r}$ the width and the height of~$R_j^\mathrm{r}$, respectively.
  Since~$v_m$ is a leaf, we have~$w_m^\mathrm{l}=w_m^\mathrm{r}=h_m^\mathrm{l}=h_m^\mathrm{r}=1$.

  We will first analyze the width and height of these rectangles for a fixed
  vertex~$v_j$, $1\le j\le m-1$. Let~$u_1^j,\ldots,u_{k_j}^j$ be the children of~$v_j$
  and denote the number of nodes in the heavy path subtrees rooted in them
  by~$n_1^j,\ldots,n_{k_j}^j$.
  Let~$n'_j=\sum_{i=1}^{k_j}n_i^j$; note that $\sum_{j=1}^m n'_j=n'-m$.
  By induction, we have~$w_i^j\le 2^{d+1}\cdot n_i^j$
  for $1\le i\le k_j$ and $h_i^j=b_i^j+t_i^j\le 2\cdot (3/2)^d \cdot n_i^j$.

  For the width of the rectangles, we obtain
  \begin{flalign*}
    w_j^\mathrm{l} = w_j^\mathrm{r}
    &\le 1+\sum_{i=1}^{k_j/2}w_i^{j*}
    =1+\sum_{i=1}^{k_j/2}\max\{w_{2i-1}^j,w_{2i}^j\}
    \le 1+\sum_{i=1}^{k_j}w_i^j\\
    &\le 1+\sum_{i=1}^{k_j}(2^{d+1}\cdot n_i^j)
    =1+2^{d+1}\sum_{i=1}^{k_j} n_i^j\\
    &\le 2^{d+1}\left(1+n'_j\right).
  \end{flalign*}

  The height of each rectangle in the construction is at least
  $2\sum_{i=1}^{k_j/2}t_i^{j*}$, but we have to add a bit more for the bottom parts of
  the boxes; in the worst case, this is $\max_{1\le \lambda\le k_j/2}b_{2\lambda-1}^j$ in~$R_j^\mathrm{l}$
  and $\max_{1\le \mu\le k_j/2}b_{2\mu}$ in~$R_j^\mathrm{r}$. Since our constructions
  ensures that $b_i^j\ge t_i^j$ , we have
  \begin{flalign*}
  h_j^\mathrm{l} + h_j^\mathrm{r}
    &\le 1+2\sum_{i=1}^{k_j/2}t_i^{j*}+\max_{1\le \lambda\le k_j/2}b_{2\lambda-1}^j+\max_{1\le \mu\le k_j/2}b_{2\sigma}^j\\
    &\le 1+2\sum_{i=1}^{k_j}t_i^j+\sum_{i=1}^{k_j}b_i^j
    \le1 + x\frac32\sum_{i=1}^{k_j} h_i^j\\
    &\le1+\frac32 \sum_{i=1}^{k_j} \left(2\cdot\left(\frac32\right)^d\cdot n_i^j\right)
    =1+2\cdot\left(\frac32\right)^{d+1} \sum_{i=1}^{k_j} n_i^j\\
    &\le 2\cdot\left(\frac32\right)^{d+1}\left(n'_j+1\right).
  \end{flalign*}

  We will now analyze the width and height of the heavy path box of the heavy path subtree
  rooted in~$v_1$. For the width, we have $w=\ell+r$ with
  $\ell\le \max_{1\le j\le m}w_j^\mathrm{l}$
  and $r\le \max_{1\le j\le m}w_j^\mathrm{r}=\ell$.
  Hence, we obtain
  \begin{flalign*}
    w&=2\ell=2\max_{1\le j\le m}w_j^\mathrm{l}\le
    2\sum_{j=1}^{m} w_j^\mathrm{l}\\
    &\le 2+2\sum_{j=1}^{m-1}\left(2^{d+1}\left(n'_j+1\right)\right)
    =2+2^{d+2}\sum_{j=1}^{m-1}\left(n'_j+1\right)\\
    &=2+2^{d+2}\cdot (n'-1)
    \le 2^{d+2}\cdot n'.
  \end{flalign*}

  For the height of the heavy path box, notice that the vertical distance between
  two vertices~$v_{j-1}$ and~$v_j$, $2\le j\le m$, is at
  most $\max\left\{h_{j-1}^\mathrm{l},h_{j}^\mathrm{r}\right\}$,
  while the vertical space below between~$v_{m}$ is at most $h_{m}^\mathrm{l}=1$.
  Hence, we have

  \begin{flalign*}
    h&=t+b\le
    h_1^\mathrm{r}+\sum_{j=2}^{m}\max\{h_{j-1}^\mathrm{l},h_{j}^\mathrm{r}\}+h_{m}^\mathrm{l}\\
    &\le \sum_{j=1}^{m}\left(h_j^\mathrm{l}+h_{j}^\mathrm{r}\right)
    = 2+\sum_{j=1}^{m-1}\left(h_j^\mathrm{l}+h_{j}^\mathrm{r}\right)\\
    &\le 2+\sum_{j=1}^{m-1}\left(2\cdot\left(\frac32\right)^{d+1} \left(n'_j+1\right)\right)
    = 2+2\cdot\left(\frac32\right)^{d+1}\left(n'-1\right)\\
    &\le 2\cdot\left(\frac32\right)^{d+1}n'.
  \end{flalign*}

  Since all heavy path trees of common depth are disjoint, the heavy-path boxes
  of common depth are as well. Further, we place only
  the top vertex of a heavy path on the boundary of its box. Finally, since we
  order the boxes such that $b^j_{2i-1}\ge b^j_{2i}$ for each~$i$ and~$j$,
  we have $h_j^\mathrm{l}\ge h_j^\mathrm{r}$ for each~$j$ and thus~$b \ge t$.
\end{proof}

Due to the properties of a heavy path decomposition, the maximum depth of any
heavy path is~$\lceil\log n\rceil$. Thus, by Lemma~\ref{lem:area},
the whole tree is drawn in a box of
width $2\cdot 2^{\lceil\log n\rceil}n=O(n^2)$ and height
$2\cdot (3/2)^{\lceil\log n\rceil}n=O(n^{1+\log 3/2})\subseteq O(n^{1.58})$.

By using the same heavy path decomposition as the algorithm of
Schulz~\cite{schulz2015} (that is, we choose the same heavy edge in case of a tie),
we have that our algorithm draws a path in~$T$ with one segment if and only if
the algorithm of Schulz draws this path with one circular arc. Hence, following
his analysis, our drawing uses at most
$\lceil3e/4\rceil=\lceil3(n-1)/4\rceil$ segments.

\begin{theorem}\label{thm:treepoly}
	Every tree admits a straight-line drawing that uses at most $\lceil3e/4\rceil$
	segments on an $O(n^2)\times O(n^{1.58})$ grid.
  This drawing can be computed in $O(n)$ time.
\end{theorem}
\begin{proof}
  The existence of the drawing is already argued above; what remains is
  to prove the time bound. By traversing the tree bottom-up, we can
  calculate the heavy path decomposition in~$O(n)$ time.
	Then, we sort the heavy paths by their depth in~$O(n)$ time using, e.g.,
  counting sort since the depth is integer and bound by $O(\log n)$.
	
	When drawing a heavy path of level~$d$,
  we only have to place the adjacent heavy-path boxes of level~$d-1$
  around the vertices of the heavy path, which gives us their coordinates
  and the corresponding heavy path box. The number of these placement
  steps is equal to the number of light edges in the graph, which is~$O(n)$.
\end{proof}

We finish this section with an adjustment of our algorithm to get a grid
drawing with the best possible number of straight-line segments.
Observe that there is only one situation in which the
previous algorithm uses more segments than necessary, that is the top
node of each heavy path. The heavy path is always drawn vertically,
while the incoming light edge of the top node will be drawn with a
different slope; aligning the two slopes would reduce the number
of segments. This suboptimality can be ``repaired'' by
\emph{tilting} the heavy path as sketched in Fig.~\ref{fig:optimaltree}.
Note that the incident subtrees with smaller depth will only be translated.
However, when drawing a heavy path box, we do not know yet with which
slope the incoming light edge will be drawn.

\begin{figure}[t]
  \centering
  \includegraphics[page=1]{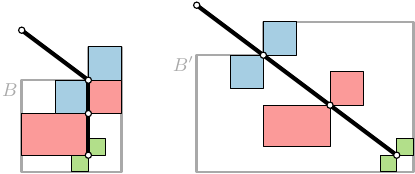}
  \caption{Further improvement on the visual complexity via increasing the size of the
boxes.}
  \label{fig:optimaltree}
\end{figure}

Our adjustment works as follows. We first draw each heavy path
vertically and then change its slope when placing the corresponding
heavy path box into a rectangle next to the parent of its top node.
Assume that we have already drawn each heavy
path subtree of depth~$d$ such that their heavy path is drawn vertically.
Let~$v$ be some vertex on a heavy path of depth~$d+1$ and
let~$B_1,\ldots,B_k$ be the heavy-path boxes with top
node~$u_1,\ldots,u_k$, respectively, adjacent to~$v$.

Note that the tilting of a heavy path only increases the right and the
bottom length of its heavy-path box; in particular, if~$B_i$ is a heavy
path box with width~$w_i=\ell_i+r_i$ and height~$h_i=t_i+b_i$ and~$B'_i$
is the tilted heavy path box with width~$w'_i=\ell'_i+r'_i$ and
height~$h'_i=t'_i+b'_i$, then we have $\ell'_i=\ell_i$, $r'_i>r_i$,
$t'_i=t_i$, and~$b'_i>b_i$. The $y$-coordinate of the
vertices~$u_1,\ldots,u_k$ after the placement depends only of the
values~$t_1\ldots,t_k$, so they will stay the same after the tilting
procedure. Hence, we can determine the $y$-coordinates by running the
previous drawing algorithm and then place the boxes in reverse order
(``inside-out'' instead of top-down). Then, the $x$-coordinate of each
vertex~$u_i$ depends only on the width of the already placed
boxes~$u_{i+1},\ldots,u_k$ and of~$r_i=r'_i$, so we know exactly by which
slope we have to tilt the box~$B_i$. Finally, to tilt the box~$B_i$,
we change the vector of each edge~$e$ on its heavy path to be an integer
multiple of the vector of the edge from~$v$ to~$u_i$ such that the vertical
length of~$e$ is not decreased. This keeps the rectangles that
contain the adjacent lower-level heavy-path boxes disjoint.

Unfortunately, this procedure increases the area of the drawing.
By tilting the heavy-path boxes, we have to blow up their size.
We are left with scaling in each ``round''
by a polynomial factor. Since there are only $\log n$ rounds, we obtain a drawing on a quasi-polynomial grid.
However, an implementation of the algorithm shows that some simple heuristics can already substantially reduce the drawing area, which gives hope that drawings on a polynomial grid exist for all trees; see Fig.~\ref{fig:implementation}.

\begin{figure}[b]
  \centering
  \subcaptionbox{The drawing produced by our algorithm.}{
    \includegraphics[page=3]{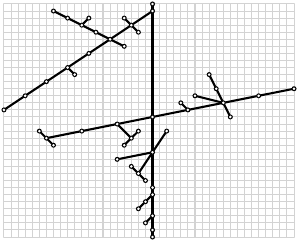}}
  \hfil
  \subcaptionbox{Heuristic improvement on the drawing area.}{
    \includegraphics[page=4]{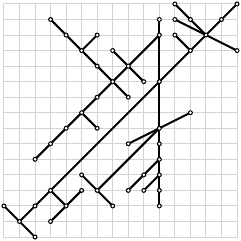}}
  \caption{Drawings of a tree with 50 vertices with minimum number of segments.}
  \label{fig:implementation}
\end{figure}

\begin{theorem}\label{thm:treequasipoly}
	Every tree admits a straight-line drawing with the smallest number of straight-line segments on a quasi-polynomial grid. This drawing can be
  computed in~$O(n)$ time.
\end{theorem}
\begin{proof}
  The existence of the drawing is argued above. The time bound is the
  same as in Theorem~\ref{thm:treepoly} as the only new step is tilting
  the heavy-path boxes, which changes the slope of each heavy edge at most
  once.
\end{proof}

\section{Planar 3-trees with few segments on the grid}\label{sect:3trees}

In this section, we show how to construct drawings of planar 3-trees with
few segments on the grid. We begin by introducing some notation.
A \emph{3-tree} is a maximal graph of treewidth~$k$, that is, no edges can
be added without increasing the treewidth. Each \emph{planar 3-tree} can be
produced from the complete graph~$K_4$ by repeatedly adding a vertex into a
triangular face and connecting it to all three vertices incident to this face.
This operation is also known as~\emph{stacking}. Any planar 3-tree admits exactly one Schnyder realizer, and it is cycle-free~\cite{brehm2000}, that is, directing all edges
according to~$T_1$, $T_2$, and~$T_n$ gives a directed acyclic graph.

\begin{figure}[b]
	\centering
	\subcaptionbox{}{\includegraphics[page=1]{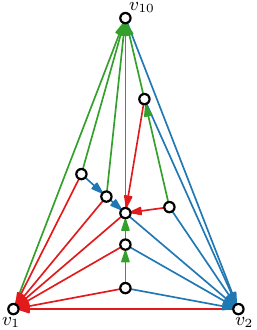}}
	\hfil
  \subcaptionbox{}{\includegraphics[page=2,trim=1.2cm 0 0 0,clip]{planar3tree}}
	\hfil
  \subcaptionbox{}{\includegraphics[page=3]{planar3tree}}
	\caption{Obtaining a canonical order from a clockwise pre-order traversal of a canonical ordering tree.
    (a) The canonical ordering trees of a planar 3-tree; (b)~the clockwise pre-order traversal on $T_2$;
    (c) the canonical order induced by this traversal.}
	\label{fig:planar3tree}
\end{figure}

Let $\mathcal T$ be a planar 3-tree.
Let $T_1$, $T_2$, and $T_n$ rooted at $v_1$, $v_2$, and $v_n$, respectively, be
the canonical ordering trees of the unique Schnyder realizer of~$\mathcal T$.
Recall that each canonical ordering tree describes a
canonical order on the vertices of~$T$ by a (counter-)clockwise pre-order
traversal~\cite{fraysseix2001}.
Without loss of generality, let~$T_1$ be the canonical ordering tree having
the fewest leaves, and let~$\sigma=(v_1,v_2,\ldots,v_n)$ be the canonical
order induced by a clockwise pre-order walk of~$T_2$; see
Fig.~\ref{fig:planar3tree}.
The fact that we choose the canonical order induced by~$T_2$ instead of~$T_1$
may seem counter-intuitive, but we will require a specific property of the
canonical order that is ensured only by this choice later.
The following lemma holds for any $v_k,4\le k\le n$.

\begin{lemma}\label{lem:contour}
  Let~$C_{k-1}=(w_1,\ldots,w_r)$ be the contour of~$G_{k-1}$,
  let~$w_p$ be the 1-parent of~$v_k$, and let~$w_q$
  be the 2-parent of~$v_k$. Then,
  \begin{enumerate}[label=(\alph*)]
    \item either $(w_q,w_p)\in T_1$ or $(w_p,w_q)\in T_2$
    \item the subgraph inside the triangle $(w_p,w_q,v_k)$ is a planar 3-tree
    \item if $q>p+1$, then $w_{p+1},\ldots,w_{q-1}$ lie inside the triangle
      $(w_p,w_q,v_k)$, $(w_{p+1},w_p)\in T_1$, and $(w_{q-1},w_q)\in T_2$.
  \end{enumerate}
\end{lemma}
\begin{proof}
  Proof for (a): For $q=p+1$ the statement holds trivially. So we assume
  that there exists a vertex~$w_o$ on~$C_k$ between~$w_p$ and~$w_q$.
  Let~$X=\{v_n\}$ and~$Y=\{w_o\}$.
  Obviously, the contour~$C_k$ is an $X$-$Y$-separator; see
  Fig.~\ref{fig:neighborsedge}. Let $Z$ be a minimal $X$-$Y$-separator that
  contains only vertices from~$C_k$. By the properties of the Schnyder
  realizer, each incoming $n$-edge of a vertex on~$C_k$ comes from a vertex
  below the contour. Hence, there is a path from each vertex on~$C_k$ to~$v_n$
  in~$T_n$ that does not revisit~$C_k$. Further, the
  paths $(w_o,v_k$), $(w_o,w_{o-1},\ldots,w_p)$, and $(w_o,w_{o+1},\ldots,w_q)$
  connect~$w_o$ to each of~$v_k$,~$w_p$ and~$w_q$ without traversing
  any other vertex from~$C_k$. Hence, $Z$ has to contain $v_k$,~$w_p$, and~$w_q$;
  otherwise, a $v_n$--$w_o$-path would remain after deleting~$Z$.
  Rose~\cite{rose1974} has shown that
  every minimal $X$-$Y$-separator in a $k$-tree is a $k$-clique; thus,
  $Z$ has to be the 3-cycle~$(w_p,w_q,v_k)$.
  The edge $(w_p,w_q)$ cannot be an $n$-edge, since otherwise we would have a cycle in~$T_n$.
  So either $(w_q,w_p)\in T_1$ or $(w_p,w_q)\in T_2$.

\begin{figure}[t]
  \centering
  \includegraphics{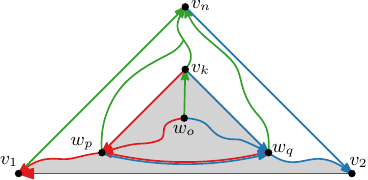}
  \caption{Proof that the edge~$(w_p,w_q)$ exists.}
  \label{fig:neighborsedge}
\end{figure}

  Proof for (b): Mondal et al.~\cite{mondal2013} have shown that this is true
  for every triangle in a planar 3-tree.

  Proof for (c): The vertices have to lie inside the triangle $(w_p,w_q,v_k)$
  since otherwise their outgoing $n$-edge to~$v_k$ would cross the
  edge~$(w_p,w_q)$. Further, since the subgraph inside~$(w_p,w_q,v_k)$ is again
  a planar 3-tree,~$w_p$, $w_q$, and~$v_k$ are the roots of the corresponding
  Schnyder realizer; $(w_{p+1},w_p)\in T_1$ and $(w_{q-1},w_q)\in T_2$
  follows immediately.
\end{proof}

Given some drawing~$\Gamma$ and two points~$a$ and~$b$ (that are not necessarily
a part of~$\Gamma$), we say that~$a$
\emph{sees}~$b$ if and only if the straight-line segment between~$a$ and~$b$ is
interior-disjoint from~$\Gamma$. If two points~$p$ and~$q$ have the same
$x$-coordinate and~$p$ lies above~$q$, then we define~$\slope(p,q)=-\infty$
and~$\slope(q,p)=+\infty$.
The following lemma was proven by Durocher and Mondal~\cite{durocher2014};
see Fig.~\ref{fig:polychain} for an illustration.

\begin{lemma}[\cite{durocher2014}]\label{lem:polychain}
  Let $a_1,\ldots,a_m$ be a strictly $x$-monotone polygonal
  chain~$C$. Let~$p$ be a point above~$C$ such that the segments~$a_1p$ and~$a_mp$
  do not intersect~$C$ except at~$a_1$ and~$a_m$. If the positive slopes of the
  edges of~$C$ are smaller than $\slope(a_1,p)$, and the negative slopes of the
  edges of~$C$ are greater than~$\slope(p,a_m)$, then every $a_i$ sees~$p$.
  \hfill\qed
\end{lemma}

\begin{figure}[b]
	\centering
	\includegraphics[page=1]{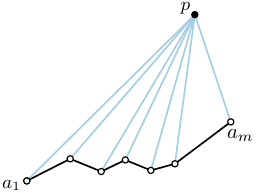}
	\caption{Illustrations of Lemma~\ref{lem:polychain}.}
	\label{fig:polychain}
\end{figure}

Note that the first condition trivially holds if all slopes of~$C$ are non-positive
and that the second condition trivially holds if all slopes of~$C$ are non-negative.

We will further need the following lemma which relies on Lemma~\ref{lem:polychain}.

\begin{lemma}\label{lem:polychain-move}
  Let $a_1,\ldots,a_m$ be a strictly $x$-monotone polygonal chain~$C$.
  Let~$p$ be a point below~$C$ that is connected to $a_1,\ldots,a_m$
  such that $\slope(a_1,p)\ge 0$ and $\slope(a_i,a_{i+1})>\slope(a_1,p)$
  for $i=1,\ldots,m-1$. Then the resulting drawing~$\Gamma$ is planar.
  Furthermore, assume that $p$ sees $a_1$ and $a_m$.
  If we move~$p$ to a point on the ray emanating
  from~$p$ with slope~$\slope(a_1,p)$, then the resulting drawing~$\Gamma'$
  preserves the embedding of~$\Gamma$ .
\end{lemma}
\begin{proof}
  We first prove that~$\Gamma$ is planar. To this end, we first mirror~$\Gamma$
  at the $x$-axis to obtain a drawing~$\Gamma^*$. Then~$a_1,\ldots,a_m$ is a strictly $x$-monotone polygonal
  chain~$C$ without positive slopes and~$p$ is a point above~$C$ such that
  all (negative) slope of~$C$ are smaller than~$\slope(a_1,p)$. Then by Lemma~\ref{lem:polychain}
  $\Gamma^*$ and, hence,~$\Gamma$ is planar.

  Now we prove that~$\Gamma'$ preserves the embedding of~$\Gamma$.
  To change the embedding, at least one triangle $(a_i,a_{i+1},p)$ has to change
  its orientation, that is, $p$ moves across the supporting line of
  some $(a_i,a_{i+1})$. Recall that for all $1\le i<m$ we have
  $\slope(a_i,a_{i+1})>\slope(a_1,p)=\slope(a_1,p')$ (see Fig.~\ref{fig:polychain-move}).
  Since we increase the distance from $a_1$ when moving from $p$ to $p'$ (and due to the slopes),
  we also increase the distance to $(a_i,a_{i+1})$. Thus,
  both~$p$ and~$p'$ lie in the same halfplane defined by the supporting line of $(a_i,a_{i+1})$.
  This proves the lemma.
\end{proof}

\begin{figure}[t]
	\centering
	\begin{subfigure}[b]{.4\linewidth}
		\centering
		\includegraphics[page=2]{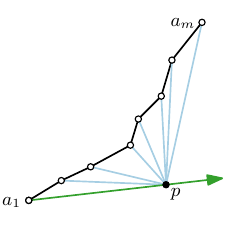}
		\caption{$\Gamma$}
	\end{subfigure}
	\hfil
	\begin{subfigure}[b]{.4\linewidth}
		\centering
		\includegraphics[page=3]{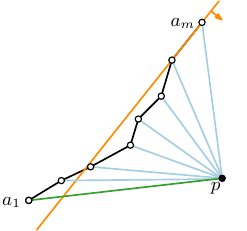}
		\caption{$\Gamma'$}
	\end{subfigure}
	\caption{Illustrations of Lemma~\ref{lem:polychain-move}.}
	\label{fig:polychain-move}
\end{figure}

\paragraph{Overview and notation.}
The main idea of the algorithm is as follows.
We inductively place the vertices according to the
canonical order~$\sigma=(v_1,\ldots,v_n)$ and refer to the step in which
vertex~$v_k$ is placed as \emph{step~$k$}. By the choice of the canonical
ordering, when we place vertex~$v_k$, its parent~$v_\ell$ in~$T_1$ has already
been placed. If~$v_k$ is the first child of~$v_\ell$ in~$T_1$, then we place~$v_k$
such that the edge~$(v_\ell,v_k)$ has the same slope as the edge from~$v_\ell$
to its parent in~$T_1$, so that both edges are drawn with one segment. This way,
we obtain a drawing of~$G$ in which the
edges of~$T_1$ are drawn with exactly one segment for every leaf of~$T_1$.

In order to be able to do this placement, we have to maintain a series of invariants
which rely on additional notation; see Fig.~\ref{fig:definitions-3tree}.
For each vertex~$v_i$, the \emph{1-out-slope} $\outl_1(i)$
is the slope of its outgoing 1-edge, the \emph{2-out-slope} $\outl_2(i)$
is the slope of its outgoing 2-edge, and the \emph{in-slope} $\inl(i)$
is the highest slope of the incoming 1-edges in the current drawing. Further, we
denote by $\parent_1(i)$ the 1-parent of~$v_i$ and by $\parent_2(i)$ the
2-parent of~$v_i$. For two vertices~$v_i,v_j$, we
denote by~$\lca(i,j)$ the lowest common ancestor of~$v_i$ and~$v_j$ in~$T_1$.
For an edge~$(v_i,v_j)$ we call the closed region bounded by~$(v_i,v_j)$, the path
$(v_i,\ldots,\lca(i,j))$, and the path $(v_j,\ldots,\lca(i,j))$ the
\emph{domain} $\dom(i,j)$ of $(v_i, v_j)$. For each step~$k$, we denote
by~$\lambda_k$ the number of leafs in the currently drawn subtree of~$T_1$, by~$s_k$ the
number of segments that are used to draw the current subtree of~$T_1$, and by $\eta_k$ the highest
slope of the 1-edges in the current drawing. We denote by~$C_k^\rightarrow$ the
part of the contour~$C_k$ between~$v_k$ and~$v_2$.

Before describing the invariants and the algorithm in detail, we first show the following
property of~$C_k^\rightarrow$.

\begin{lemma}\label{lem:contourright}
  The edges on $C_k^\rightarrow$ are exactly the path from~$v_k$ to~$v_2$ in~$T_2$,
  and $\parent_2(k)\in C_{k-1}^\rightarrow$.
\end{lemma}
\begin{proof}
  We proof the lemma by induction. For~$k=3$, we have
  that~$C_k^\rightarrow=(v_3,v_2)$ with~$\parent_2(3)=v_2$, so the lemma
  holds. For~$k>3$, recall that we chose the canonical order induced
  by~$T_2$ in clockwise pre-order. Hence, either~$v_{k-1}$ is the parent
  of~$v_k$ in~$T_2$, or they have a common ancestor~$v_\ell$ with~$\ell<k-1$
  in~$T_2$. In the former case,~$\parent_2(k)=v_{k-1}\in C_{k-1}^\rightarrow$
  and $C_k^\rightarrow=v_k\circ C_{k-1}^\rightarrow$, so the lemma holds by
  induction. In the latter case, we in fact have $\parent_2(k)=v_\ell$.
  By induction,~$C_{k-1}^\rightarrow=(v_{k-1},\ldots,v_\ell,\ldots,v_2)$,
  so~$\parent_2(k)\in C_{k-1}^\rightarrow$ and
  $C_k^\rightarrow=(v_k,v_\ell,\ldots,v_2)$.
\end{proof}

\begin{figure}[t]
  \centering
  \includegraphics{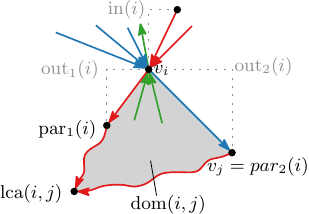}
  \caption{Definitions for the drawing algorithm for planar 3-trees.}
  \label{fig:definitions-3tree}
\end{figure}

\paragraph{Invariants.}
After each step~$k\ge 3$, we maintain the following invariants.

\begin{enumerate}[label=(I\arabic*)]
	\item\label{inv-chain} The contour~$C_k$ is a strictly $x$-monotone polygonal
    chain; the $x$-coordinates along~$C_k^\rightarrow$ increase by exactly~1 per vertex.
  \item\label{inv-segments} The 1-edges are drawn with $s_k=\lambda_k$ segments
    in total with integer slopes between~$1$ and~$\eta_k\le\lambda_k$.
  \item\label{inv-domain} For each~$(v_i,v_j)\in C_k^\rightarrow$
    and for each 1-edge~$e\neq(\parent_1(j),v_j)$ in $\dom(i,j)$ it holds that
    $\slope(e)>\outl_1(j)$.
  \item\label{inv-contour} For each~$v_i\in C_k^\rightarrow$, $\outl_1(i)>\outl_2(i)$.
  \item\label{inv-planar} The current drawing is crossing-free and
    for each~$(v_i,v_j)\in T_n$, we have $\slope(v_i,v_j)>\outl_1(j)$.
  \item\label{inv-area} Vertex~$v_1$ is placed at coordinate~$(0,0)$,~$v_2$ is
    placed at coordinate~$(k-1,0)$, and every vertex lies inside the
    rectangle $(0,0)\times(k-1,(k-1)\lambda_k)$.
\end{enumerate}

\paragraph{The algorithm.}
The algorithm starts with placing~$v_1$ at~$(0,0)$, $v_2$ at~$(2,0)$, and~$v_3$
at~$(1,1)$. Obviously, all invariants~\ref{inv-chain}--\ref{inv-area} hold.
In step~$k>3$, the algorithm proceeds in two steps. Recall that~$v_k$ is a
neighbor of all vertices on the contour between~$v_l=\parent_1(k)$ and~$v_r=\parent_2(k)$.

\subparagraph{Step 1: Insertion.} In the insertion step,~$v_k$ is placed with the same
$x$-coordinate as~$v_r$.  We distinguish between three cases to obtain the
$y$-coordinate of~$v_k$; see Fig.~\ref{fig:opinsert} for an illustration.
\begin{enumerate}[label=(\roman*)]
  \item\label{opinsert-1} If no incoming 1-edge of~$v_l$ has been drawn
    yet, then we draw the edge~$(v_l,v_k)$ with
    slope~$\outl_1(l)$.
  \item\label{opinsert-2} If~$v_l$ already has an incoming 1-edge
    and~$v_l$ and~$v_r$ are the only
    neighbors of~$v_k$ in the current drawing, then we draw the
    edge~$(v_l,v_k)$ with slope $\inl(l)+1$.
  \item\label{opinsert-3} Otherwise, we draw the edge~$(v_l,v_k)$ with
    slope~$\eta_{k-1}+1$.
\end{enumerate}
Note that this does not maintain invariant~\ref{inv-chain}.

\begin{figure}[t]
  \centering
  \begin{subfigure}[t]{.3\textwidth}
    \centering
    \includegraphics[page=1]{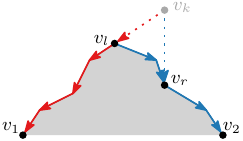}
    \caption{Case~\ref{opinsert-1}}
    \label{fig:opinsert-1}
  \end{subfigure}
  \hfill
  \begin{subfigure}[t]{.3\textwidth}
    \centering
    \includegraphics[page=2]{opinsert}
    \caption{Case~\ref{opinsert-2}}
    \label{fig:opinsert-2}
  \end{subfigure}
  \hfill
  \begin{subfigure}[t]{.3\textwidth}
    \centering
    \includegraphics[page=3]{opinsert}
    \caption{Case~\ref{opinsert-3}}
    \label{fig:opinsert-3}
  \end{subfigure}
  \caption{Inserting vertex~$v_k$ while maintaining invariant~\ref{inv-planar}.}
  \label{fig:opinsert}
\end{figure}

\subparagraph{Step 2: Shifting.} In the shifting step, the vertices
between~$v_r$ and~$v_2$ on the contour~$C_k$ have to be shifted to the
right without increasing the number of segments~$s_k$ used to draw~$T_1$. To
this end, we iteratively extend the outgoing 1-edge of these vertices, starting
with~$v_r$, to increase their $x$-coordinates all by~$1$; see
Fig.~\ref{fig:opdiagshift}. This procedure places the vertices on the grid since
the slopes of the extended edges are all integer by invariant~\ref{inv-segments}.

\begin{figure}[b]
  \centering
  \includegraphics[page=1]{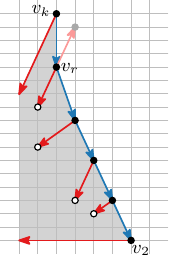}
  \hfill
  \includegraphics[page=2]{opdiagshift}
  \hfill
  \includegraphics[page=3]{opdiagshift}
  \hfill
  \includegraphics[page=4]{opdiagshift}
  \caption{Shifting $v_r,\ldots,v_2$ along their outgoing 1-edge.}
  \label{fig:opdiagshift}
\end{figure}

In the following, we will show that every planar 3-tree admits a straight-line
drawing that uses at most $(8n-17)/3$ segments on an $O(n)\times O(n^2)$ grid,
and that this drawing can be computed in $O(n^2)$ time.
We have to show that, after each step~$k>3$, all
invariants~\ref{inv-chain}--\ref{inv-area} hold. We denote by $x(v)$ the $x$-coordinate
of the vertex~$v$, and similarly, by $y(v)$ the $y$-coordinate of~$v$.

\paragraph{Invariant~\ref{inv-chain}.}
The contour in step~$k$ is $C_k=(v_1,\ldots, v_l, v_k, v_r,\ldots,v_2)$.
By induction, $(v_1,\ldots,v_l,v_r,\ldots,v_2)$ is strictly $x$-monotone.
We place~$v_k$ above~$v_r$ and move all vertices from~$v_r$ to~$v_2$
to the right by one unit, so~$x(v_l)<x(v_k)=x(v_r)-1$ and~$C_k$ is strictly
$x$-monotone. By Lemma~\ref{lem:contourright}, $v_r$ lies on~$C_{k-1}^\rightarrow$,
so the $x$-coordinates along~$(v_r,\ldots,v_2)$ increase by~1 per vertex;
since~$x(v_r)=x(v_k)$, the same holds for~$C_k^\rightarrow$.

\paragraph{Invariant~\ref{inv-segments}.} In case~\ref{opinsert-1} of the insertion step,
$v_l$ changes from a leaf to a non-leaf, while~$v_k$ is added as a
leaf. Since~$(v_l,v_k)$ is drawn with the same slope as the outgoing
1-edge of~$v_l$, we have $s_k=s_{k-1}=\lambda_{k-1}=\lambda_k\ge \eta_k$.
In cases~\ref{opinsert-2} and~\ref{opinsert-3} of the insertion step, $v_l$ was not a leaf
before and a new integer slope is used to draw~$(v_l,v_k)$, so we have
$s_k=s_{k-1}+1=\lambda_{k-1}+1=\lambda_k$; since the maximum slope increases
by at most one, we have $\eta_k\le\eta_{k-1}+1\le\lambda_k$. For the shift step,
as the vertices on~$C_k^\rightarrow$ have no incoming 1-edges, and the slopes of their
outgoing 1-edges do not change, the number of segments remains the same.

\paragraph{Invariant~\ref{inv-domain}.}
We have to only address the insertion step, since the shifting step
does not affect the relevant slopes.
Since the slopes of the 1-edges do not change,
and since the edges on $(v_r,\ldots,v_2)$ lie on~$C_{k-1}^\rightarrow$
by Lemma~\ref{lem:contourright}, the invariant holds for edges on $(v_r,\ldots,v_2)$ by induction;
hence, it suffices to show the invariant for $(v_k,v_r)$.
By Lemma~\ref{lem:contour}(a), the edge $(v_l,v_r)$ exists in~$T_1$ or~$T_2$.
We distinguish between two cases.

Case 1: $(v_r,v_l)\in T_1$; see Fig.~\ref{fig:domain-1}. It immediately follows that
$\lca(k,r)=v_l=\parent_1(k)$. Hence, the domain $\dom(k,r)$ is bounded  by the
triangle~$(v_l,v_r,v_k)$, which is a planar 3-tree by Lemma~\ref{lem:contour}(b)
with~$v_l$ as the root of~$T_1$. By construction, the
edge~$(v_r,v_l)$ has the smallest slope of all incoming 1-edges of~$v_l$ within
the domain. In particular, by construction, the first incoming 1-edge of each vertex
is assigned the same slope as the outgoing 1-edge, while all other incoming
1-edges are assigned higher slopes. Thus, all 1-edges in the domain must have
higher slopes than $(v_r,v_l)$.

  \begin{figure}[t]
    \centering
    \begin{subfigure}[t]{.48\linewidth}
      \centering
      \includegraphics[page=1]{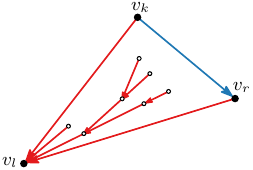}
      \caption{Case $(v_r,v_l)\in T_1$}
      \label{fig:domain-1}
    \end{subfigure}
    \hfill
    \begin{subfigure}[t]{.48\linewidth}
      \centering
      \includegraphics[page=2]{domain}
      \caption{Case $(v_l,v_r)\in T_2$}
      \label{fig:domain-2}
    \end{subfigure}
    \caption{The domain $\dom(k,r)$ as in invariant~\ref{inv-domain}.}
    \label{fig:domain}
  \end{figure}

Case 2: $(v_l,v_r)\in T_2$; see Fig.~\ref{fig:domain-2}. Consider the
unique path from $v_k$ to $v_1$ in~$T_1$. Denote the vertices on this path
by $t_0,t_1,t_2,\ldots,t_q$ with $t_0=v_k$, $t_1=v_l$, and $t_q=v_1$.
By Lemma~\ref{lem:contour}, $t_1=v_l$ is connected to~$v_r$ and~$(t_1,v_r)\in T_1$
or $(t_1,v_r)\in T_2$. If~$(t_1,v_r)\in T_2$, then consider the drawing of~$G_l$
after step~$l$ of the algorithm. In this step, $t_1=v_l$ was placed,~$v_r$
is its $2$-parent, and~$t_2$ is its $1$-parent. Hence, by Lemma~\ref{lem:contour},
$t_2$ is connected to~$v_r$ and~$(t_2,v_r)\in T_1$ or $(t_2,v_r)\in T_2$.
It follows by induction that, as long as~$(t_i,v_r)\in T_2$, then
$(t_{i+1},v_r)$ exists.

Let~$t_s$ be the first vertex on the path $t_0,t_1,\ldots,t_q$
such that~$(t_s,v_r)$ exists but is not a 2-edge. Since~$t_q=v_1$ but~$v_1$ has no $2$-parent,
this vertex exists; since
$(t_0=v_k,v_r),(t_1=v_l,v_r)\in T_2$, we have~$s>1$. By choice of~$t_s$, we have
that~$(t_{s-1},t_s)\in T_1$ and~$(t_{s-1},v_r)\in T_2$; hence, by
Lemma~\ref{lem:contour}(a), the edge~$(v_r,t_s)$ exists and since $(t_s,v_r)\not\in T_2$
it has to be a 1-edge.
Hence, we have~$\parent_1(r)=t_s=\lca(k,r)$. The domain~$\dom(k,r)$ is bounded
by the edges~$(v_k,v_r)$,~$(v_r,t_s)$, and the path~$(t_0:=v_k,t_1,\ldots,t_s)$.
By Lemma~\ref{lem:contour}(b), the domain can be divided into~$s$ planar 3-trees
$(t_i,v_r,t_{i-1})$, $1\le i\le s$. For the triangle~$(t_s,v_r,t_{s-1})$, the
argument from Case 1 can be directly applied; for~$i<s$, we can use the same
argument to show that no 1-edge has a higher slope than~$(t_i,t_{i+1})$,
which proves the invariant.

\paragraph{Invariant~\ref{inv-contour}.} Let~$v_i\in C_k^\rightarrow$. First, consider the
case that~$v_i=v_k$; see Fig.~\ref{fig:contourslopes-1}.
In the insertion step,~$v_k$ is placed at the same
$x$-coordinate as~$v_r$. In the shifting step,~$v_r$ is moved to the right by
one column and upwards by $\outl_1(r)$ rows; hence, we have~$y(v_r)<y(v_k)+\outl_1(r)$.
Since $\dom(k,r)$ contains~$v_l=\parent_1(k)$, it follows from invariant~\ref{inv-domain}
that $\outl_2(k)=y(v_r)-y(v_k)<\outl_1(r)<\outl_1(k)$.

  \begin{figure}[t]
    \centering
    \begin{subfigure}[t]{.46\linewidth}
      \centering
      \includegraphics[page=1]{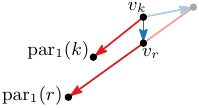}
      \caption{$i=k$ and $v_r=\parent_2(k)$}
      \label{fig:contourslopes-1}
    \end{subfigure}
    \hfil
    \begin{subfigure}[t]{.46\linewidth}
      \centering
      \includegraphics[page=2]{contourslopes}
      \caption{$i\neq k$ and $v_j=\parent_2(i)$}
      \label{fig:contourslopes-2}
    \end{subfigure}
    \caption{Shifting $\parent_2(i)$ while maintaining invariant~\ref{inv-contour}.}
    \label{fig:contourslopes}
  \end{figure}

Now, consider the case that~$v_2\neq v_i\neq v_k$; see
Fig.~\ref{fig:contourslopes-2}. Let~$v_j=\parent_2(i)$.
By Lemma~\ref{lem:contourright}, $(v_i,v_j)\in C_{k-1}^\rightarrow$, so
we had $\outl_1(i)>\outl_2(i)$ before the shifting step of the
algorithm. In the shifting step, $\outl_1(i)$ does not change. However,~$v_i$
and~$v_j$ are shifted to the right by one column and upwards by~$\outl_1(i)$
rows and~$\outl_1(j)$ rows, respectively; hence,~$\outl_2(i)$ increases
by~$\outl_1(j)-\outl_1(i)$. Since~$\dom(i,j)$ contains~$v_l=\parent_1(i)$, it
follows from invariant~\ref{inv-domain} that $\outl_1(j)<\outl_1(i)$. This
implies that~$\outl_2(i)$ becomes smaller by the shifting step,
so~$\outl_1(i)>\outl_2(i)$ is maintained.

\paragraph{Invariant~\ref{inv-planar}.}
First, consider the drawing after the insertion step, but before the shifting
step. By induction, the drawing was crossing-free before this step.
Since no existing edge is modified, each crossing has to involve an edge
of~$v_k$. We will use Lemma~\ref{lem:polychain} to show that~$v_k$ sees
all its neighbors.
By invariant~\ref{inv-chain}, the neighbors of~$v_k$ lie on a strictly
$x$-monotone polygonal chain. Because of~$\slope(v_k,v_r)=-\infty$, it suffices
to show that the edges between~$v_l$ and~$v_r$ on the contour have smaller slope
than~$\slope(v_l,v_k)$.

Consider case~\ref{opinsert-1} of the insertion step:~$v_l$ has no
incoming 1-edge; see Fig.~\ref{fig:opinsert-1}. Then, $(v_l,v_r)\in C_{k-1}$; otherwise, we get a
contradiction from Lemma~\ref{lem:contour}(c). Hence,~$v_l$ and~$v_r$ are
the only neighbors of~$v_k$ and we have to show only
that~$\slope(v_k,v_l)>\slope(v_l,v_r)$. By construction and
invariant~\ref{inv-contour}, we have
$\slope(v_k,v_l)=\outl_1(l)>\outl_2(l)=\slope(v_l,v_r)$.

Consider case~\ref{opinsert-2} of the insertion step: $v_l$ already has an
incoming 1-edge and $v_l$ and~$v_r$ are the only neighbors of~$v_k$; see Fig.~\ref{fig:opinsert-2}.
In particular, that means that~$(v_r,v_l)\in T_1$ and, by construction,
$\slope(v_k,v_l)=\inl(l)+1>\inl(l)\ge\slope(v_r,v_l)$.

Consider case~\ref{opinsert-3} of the insertion step: $v_l$ already has an
incoming 1-edge and $(v_l,v_r)\notin C_{k-1}$; see Fig.~\ref{fig:opinsert-3}.
In this case,~$(v_k,v_l)$ is
drawn with $\slope(v_k,v_l)=\eta_{k-1}+1$, so every 1-edge on the contour
between~$v_l$ and~$v_r$ has lower slope than~$(v_k,v_l)$. By
invariant~\ref{inv-segments} and~\ref{inv-contour}, the 2-edges on the contour between~$v_l$ and~$v_r$
also cannot have higher slope than $\eta_{k-1}$.

Hence, we can use Lemma~\ref{lem:polychain} in each case of the insertion step
to show that the drawing remains planar. The second part of the invariant,
for each~$(v_i,v_j)\in T_n$, $\slope(v_i,v_j)>\outl_1(j)$, holds for each
vertex but~$v_k$ by induction. Since the edges of~$v_k$ are planar
and~$v_k$ is placed vertically above~$v_r=\parent_2(k)$, all
edges~$(v_i,v_k)\in T_n$ have positive slope
with~$x(v_i)>x(v_l)$, so $\slope(v_i,v_k)>\outl_1(k)$.

Now, consider the drawing after the shifting step. We shift the
vertices between~$v_r$ and~$v_2$ one by one on the contour~$C_k^\rightarrow$ and show that,
after each step, the invariant holds. Let~$v_t$ be the vertex next to shift,
let~$v_s$ be its predecessor and let~$v_u=\parent_2(t)$ be its successor on the contour.
Since either~$v_s=v_k$ or~$v_s$ has been shifted in the previous step, the
edge~$(v_s,v_t)$ is drawn vertically and the edge~$(v_t,v_u)$ has
slope~$\outl_2(t)$.
Since~$(v_t,\parent_1(t))$ lies on the boundary of both domains~$\dom(s,t)$
and~$\dom(t,u)$, the union of these domains contains all edges incident
to~$v_t$ and all $n$-edges of~$v_t$ lie in~$\dom(t,u)$. We will now show
that both domains remain planar and for all $(v_i,v_t)\in T_n$,
$\slope(v_i,v_t)>\outl_1(t)$.

\begin{figure}[t]
  \centering
  \begin{subfigure}[b]{.45\textwidth}
    \centering
    \includegraphics[page=1]{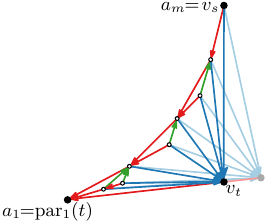}
    \caption{Maintaining the invariant in~$\dom(s,t)$}
    \label{fig:planar-diagshift-1}
  \end{subfigure}
  \hfil
  \begin{subfigure}[b]{.45\textwidth}
    \centering
    \includegraphics[page=2]{planar-diagshift}
    \caption{Maintaining the invariant in~$\dom(t,u)$}
    \label{fig:planar-diagshift-2}
  \end{subfigure}
  \caption{Shifting vertex~$v_t$ while maintaining invariant~\ref{inv-planar}.}
  \label{fig:planar-diagshift}
\end{figure}

Consider the domain~$\dom(s,t)$; see Fig.~\ref{fig:planar-diagshift-1}. By the Schnyder realizer properties, all
edges incident to~$v_t$ (except~$(v_t,\parent_1(t))$) are 2-edges.
Let~$\parent_1(t)=a_1,\ldots,a_m=v_s$ be the path through the neighbors
of~$v_t$. All edges on this path are 1-edges or $n$-edges (otherwise~$T_2$
would not be a tree). By invariant~\ref{inv-domain}, all 1-edges on this
path have slope higher than~$\outl_1(t)$. By induction and by
invariant~\ref{inv-domain}, for each $n$-edge
$(a_{i-1},a_i)$ on this path, we
have~$\slope(a_{i-1},a_i)>\outl_1(i)>\outl_1(t)$. As $\outl_1(t)>0$,
the path~$a_1,\ldots,a_m$ is an $x$-monotone chain without negative slopes.
Hence, we can apply Lemma~\ref{lem:polychain-move} to prove that the new
position of~$v_t$ preserves the embedding, so the drawing remains planar.

Now, consider the domain~$\dom(t,u)$; see Fig.~\ref{fig:planar-diagshift-2}. By the Schnyder realizer properties, all
edges incident to~$v_t$ (except~$(v_t,\parent_2(t))$) are $n$-edges.
Let~$\parent_1(t)=a_1,\ldots,a_m=v_u$ be the path through the neighbors
of~$v_t$. All edges on this path are 1-edges or 2-edges (otherwise~$T_n$
would not be a tree). If~$m=2$, then there is no $n$-edge in~$\dom(t,u)$.
By invariant~\ref{inv-domain}, we have~$\outl_1(t)>\slope(a_1,a_2)=\outl_1(u)$
and  we can move~$v_t$ in
direction~$\overrightarrow{\parent_1(t) v_t}$ without changing the embedding.
If~$m>2$, we look back on step~$t$ of the algorithm. Observe
that~$a_1,\ldots,a_m$ were the neighbors of~$v_t$ on the contour~$C_{t-1}$.
Since case~\ref{opinsert-1} and~\ref{opinsert-2} imply~$m=2$,~$v_t$ had to be
inserted with case~\ref{opinsert-3}. By construction, the
edge~$(\parent_1(t),v_t)$ was inserted with a higher slope than all 1-edges and,
by invariant~\ref{inv-domain}, also all 2-edges on the contour. Further,
as~$a_1,\ldots,a_{m-1}$ were removed from the contour and the algorithm only
changes the drawing of edges incident to a vertex on the contour, the edges~$(a_i,a_{i+1})$
for $1\le i\le m-2$ were not changed afterward. Thus, we still have
$\outl_1(t)>\slope(a_i,a_{i+1}),1\le i\le m-2$. Additionally, by
invariant~\ref{inv-domain}, also $\outl_1(t)>\slope(a_{m-1},a_m)$ holds, so by
Lemma~\ref{lem:polychain} we can place~$v_t$ above~$v_u$ with slope~$\outl_1(t)$
without crossings (since we obtain~$\slope(v_t,v_u)=-\infty$).
The movement of~$v_t$ to the right and the planarity maintain
that~$\slope(a_i,v_t)>\outl_1(t),2\le i\le m-1$.
This establishes invariant~\ref{inv-planar}.

\paragraph{Invariant~\ref{inv-area}.}
Since~$v_1$ is never moved, it remains at~$(0,0)$. Vertex~$v_2$ is moved $k-3$ times to
the right by one column, so it is located at~$((k-3)+2,0)=(k-1,0)$.
As~$v_k$ is not placed to the right of~$v_2$ and no vertex is moved to the
right by more than one column, the maximum $x$-coordinate of all vertices
is~$k-1$. The contour consists of only 1-edges and 2-edges; by
invariants~\ref{inv-contour} and~\ref{inv-segments}, their maximum slope
is~$\lambda_k$. Hence, the maximum $y$-coordinate of all vertices is at most~$(k-1)\lambda_k$.

This proves the correctness of all invariants. We will now analyze the number
of segments used by the algorithm. By Invariant~\ref{inv-segments}, the 1-edges
are drawn with~$s_n=\lambda_n$ slopes, where~$\lambda_n$ is equals to the
number of leaves in~$T_1$. Recall that the number of leaves in the Schnyder
realizer is at most~$2n-5$~\cite{bonichon2002} and that we chose~$T_1$ as the
canonical ordering tree with the fewest leaves, so we
have~$\lambda_n\le(2n-5)/3$. The trees~$T_2$
and~$T_n$ contain~$n-1$ vertices, so they are both drawn with at most~$n-2$
segments. Hence, the total number of segments is at
most~$(2n-5)/3+(n-2)+(n-2)=(8n-17)/3$.

It remains to show the running time of the algorithm.
Both the Schnyder realizer and the canonical order can be computed in
linear time~\cite{schnyder1990,dpp-hdpgg-C90}. The insertion step can
be handled in constant time by storing at each vertex the highest slope
of its incoming 1-edges (if it exists) and updating the value~$\eta_k$
at each step. In the shifting step, we have to move all vertices on the
contour between the placed vertex and~$v_2$ by stretching their outgoing
1-edges. Since there are at most~$n-2$ vertices on this part of the
contour and we have~$n-2$ steps, this takes~$O(n^2)$ time in total.
It is not clear whether we can use an approach similar to the one by
Chrobak and Payne~\cite{cp-ltadp-ipl95} in order to reduce the running
time of the shifting step as the vertices have to be moved in two
directions instead of only one and since we have to know the exact
coordinates and slopes at every step.

This proves the following theorem, which is the main result of this section.
See Figure~\ref{fig:3tree-example} for an example.

\begin{figure}
  \centering
  \includegraphics[page=5]{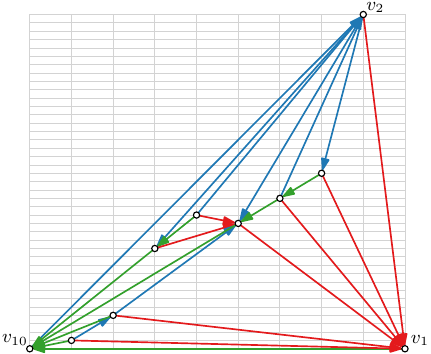}
  \caption{The drawing produced by our algorithm for the planar 3-tree depicted
  in Figure~\ref{fig:planar3tree}.}
  \label{fig:3tree-example}
\end{figure}

\begin{theorem}\label{thm:planar3tree}
  Every planar 3-tree admits a straight-line
  drawing that uses at most $(8n-17)/3$ segments on an $O(n)\times O(n^2)$ grid. This drawing can be computed in $O(n^2)$ time.
\end{theorem}

\section{Maximal outerplanar graphs with segments on the grid}\label{sect:outerplanar}

A graph is \emph{outerplanar} if it can be embedded in the plane with all vertices
on one face (called outerface), and it is \emph{maximal outerplanar} if no edge can be added
while preserving outerplanarity. This implies that all interior faces
of a maximal outerplanar graph are triangles.
Outerplanar graphs have
degeneracy~2~\cite{lick1970}, that is, every induced subgraph of an outerplanar graph
has a vertex with degree at most two. Thus, we find in every maximal outerplanar graph
a vertex of degree~2 whose removal (taking away one triangle) results in another
maximal outerplanar graph.
By this, we gain a deconstruction
order (also known as ear-decomposition) that stops with a triangle.
Let~$G=(V,E)$ be a maximal outerplanar graph and
let~$\sigma=(v_1,\ldots,v_n)$ be the reversed deconstruction order.

\begin{lemma}\label{lem:outertrees}
  The edges of~$G$ can be partitioned into two trees~$T_1$ and~$T_2$. Moreover,
  we can turn~$G$ into a planar 3-tree by adding a vertex and edges in the outerface.
  The additional edges form a tree~$T_n$. The three trees~$T_1$, $T_2$ and~$T_n$
  induce a Schnyder realizer.
\end{lemma}
\begin{proof}
  We build the graph~$G$ according to the reversed deconstruction order.
  Let~$G_k$ denote the subgraph of~$G$ induced by the set $\{v_1,v_2,\ldots, v_k\}$.
  Further, let~$G'_k$ be the graph obtained by adding the vertex~$v_n$ and
  the edges $(v_i,v_n)$ for all $1 \le i \le k$ to~$G$.
  We prove by induction over~$k$ that there exists a Schnyder
  realizer induced by $T_1$, $T_2$, $T_n$ for~$G'_k$ 	
  such that the trees $T_1$ and $T_2$ form the graph~$G_k$ and~$G'_k$ is a
  planar 3-tree. Note that Felsner and Trotter~\cite{felsner2005} already proved
  this lemma without the statement that $G'_k$ is a planar 3-tree.

  For the base case $k=2$, our hypothesis is certainly true; see~$G'_2$ in
  Fig.~\ref{fig:outertrees}. Assume our assumption holds for
  some~$k$. In order to obtain~$G_{k+1}$ from $G_{k}$, we have to add
  the vertex~$v_{k+1}$ and two incident edges~$(v_i,v_{k+1})$ and $(v_j,v_{k+1})$.
  Assume that~$v_i$ is left of~$v_j$ on~$C_k$. We add $(v_i,v_{k+1})$ to~$T_1$ and $(v_j,v_{k+1})$
  to~$T_2$. This is safe since we cannot create a cycle in any of the trees.
  There is another a new edge $(v_{k+1},v_n)$ in $G'_{k+1}$ which we add to $T_n$.
  Again, no cycle can be created. The three outgoing tree edges at~$v_i$
  form three wedges (unless $v_i$ is the root of $T_1$).
  The new ingoing 1-edge $(v_i,v_{k+1})$ lies in the wedge bounded by
  the outgoing 2-edge and the outgoing $n$-edge. For the edge~$(v_j,v_{k+1})$, we
  can argue analogously. Hence, the three trees induce a Schnyder realizer; see also
  Fig.~\ref{fig:outertrees}.
  Moreover, the graphs~$G'_k$ and~$G'_{k+1}$ differ exactly by the vertex $v_{k+1}$
  that has been stacked into a triangular face of~$G'_k$; thus, since $G'_k$ is
  a planar 3-tree, so is~$G'_{k+1}$. We now have proven the induction hypothesis
  for $k+1$. To obtain the statement of the lemma, we take the Schnyder realizer
  for the graph~$G'_{n-1}$ and move the edge~$(v_1,v_n)$ from~$T_n$ to~$T_1$ and
  the edge~$(v_2,v_n)$ from~$T_n$ to~$T_2$. Now~$T_1$ and~$T_2$ form~$G$, and
  all three trees induce the Schnyder realizer of a planar 3-tree.
\end{proof}

\begin{figure}[t]
	\centering
\includegraphics{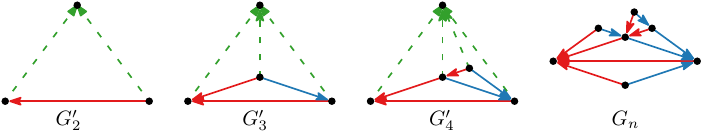}
\caption{Construction of an maximal outerplanar graph embedded in a planar 3-tree as done in the proof of Lemma~\ref{lem:outertrees}. The tree~$T_n$ is dashed.}
\label{fig:outertrees}
\end{figure}

We can now rely on our methods developed for the planar 3-trees.
The technique used in Theorem~\ref{thm:planar3tree}
produces a drawing of a planar 3-tree in which two of the trees of the Schnyder realizer
are drawn with single-edge segments, whereas the third tree uses as many segments
as it had leaves.

Consider a drawing according to Theorem~\ref{thm:planar3tree} of the planar
3-tree introduced in Lemma~\ref{lem:outertrees} that
contains the maximal outerplanar graph~$G$ as an induced subgraph.
By deleting~$T_n$, we obtain a drawing of~$G$. Note that in this
drawing the outer face is realized as an interior face, which can be
avoided (if this is undesired) by repositioning~$v_n$ accordingly.
The Schnyder realizer has $2n-5$ leaves in total, but $n-3$ of them
belong to~$T_n$. We can assume that $T_1$ has the smallest number of leaves, which
is at most $n/2-1$. We need $n-2$ segments for drawing~$T_2$ (one per edge), and three
edges for the triangle~$v_1,v_2,v_n$.
In total, we have at most $n/2-1+n-2 + 3= 3n/2$ segments. Since the drawing
is a subdrawing from our drawing algorithm for planar 3-trees, we get the same area
bound and running time as in the planar 3-tree scenario. We summarize
our results in the following theorem.

\begin{theorem}\label{thm:outerplanar}
    Every maximal outerplanar graph admits a straight-line drawing that uses at
    most~$3n/2$ segments on an $O(n)\times O(n^2)$ grid.
     This drawing can be computed in $O(n^2)$ time.
\end{theorem}

\section{Triangulations with circular arcs}\label{sect:triangulations}

In this section, we present an algorithm to draw a triangulation using few circular arcs.
An embedded graph is a \emph{triangulation} (or maximal planar) if every face is a 3-cycle, including the outerface.
Our algorithm draws upon ideas for drawing triangulations with line segments by Durocher and Mondal \cite{durocher2014} as well as Schulz's algorithm for drawing 3-connected planar graphs with circular arcs~\cite{schulz2015}.

\begin{figure}[b]
\centering
\begin{subfigure}[b]{.45\textwidth}
  \centering
  \includegraphics[page=1]{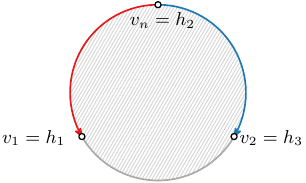}
  \caption{Initial state of the algorithm.}
  \label{fig:initialstate-1}
\end{subfigure}
\hfil
\begin{subfigure}[b]{.5\textwidth}
  \centering
  \includegraphics[page=2]{initialstate}
  \caption{The state after processing~$v_n$.}
  \label{fig:initialstate-2}
\end{subfigure}
\caption{Illustration of the algorithm to draw triangulations with few circular arcs. Hatching indicates the undrawn region.}
\label{fig:initialstate}
\end{figure}

Similar to Schulz~\cite{schulz2015}, a canonical order $v_1, \ldots, v_n$ on the vertices of a triangulation is reversed and used to structure our drawing algorithm.
We start by drawing $v_1$, $v_2$, and $v_n$ on a circle; see Fig.~\ref{fig:initialstate-1}. We assume that they are placed as shown and hence refer to the arc connecting $v_1$ and $v_2$ as \emph{the bottom arc}.
The interior of the circle is the \emph{undrawn} region~$\mathcal{U}$ which we maintain as a strictly convex shape.
The vertices incident to $\mathcal{U}$ are referred to as the \emph{horizon} and denoted $h_1, h_2, \ldots, h_{k-1}, h_k$ in order; see Fig.~\ref{fig:initialstate-2}. We maintain that~$h_1 = v_1$ and $h_k = v_2$.
Initially, we have $k = 3$ and $h_2 = v_n$.
We iteratively take a vertex~$h_i$ of the horizon (the latest in the canonical order) to process it.
That is, we draw its undrawn neighbors and edges between these, thereby removing $h_i$ from the horizon.
Note that $h_i$ will never be the first or last vertex on the horizon ($v_1$ or~$v_2$).

\paragraph{Invariant.}
We maintain as invariant that each vertex~$v$ (except $v_1$, $v_2$, and~$v_n$) has a
segment $\ell_v$ incident from above such that its downward extension intersects
the bottom arc strictly between $v_1$ and $v_2$.
Since $\mathcal{U}$ is strictly convex, observe that for any vertex $h$ on the horizon
the only intersection points of the line through $\ell_h$ with the undrawn region's boundary
are this intersection with the bottom arc and~$h$ itself.

\paragraph{Processing a vertex.}
To process a vertex $h_i$, we first consider the triangle $h_{i-1}h_ih_{i+1}$: this triangle (except for its corners) is strictly contained in $\mathcal{U}$.
We draw a circular arc $A$ from $h_{i-1}$ to $h_{i+1}$ with maximal curvature, but within this triangle; see Fig.~\ref{fig:determinearc-1}.
This ensures a plane drawing and maintains a strictly convex undrawn region.
Moreover, it ensures that $h_i$ can ``see'' the entire arc~$A$.

Vertex $h_i$ may have a number of neighbors that were not yet drawn.
We dedicate a fraction of the arc $A$ to placing these neighbors, .
In particular, this fraction is determined by the intersections of segments $v_1h_i$ and $v_2h_i$ with $A$; see Fig.~\ref{fig:determinearc-2}. By convexity of~$\mathcal{U}$, these intersections exist.
If $h_{i-1}$ is equal to $v_1$, then the intersection for $v_1h_i$ degenerates to $v_1$; similarly, the intersection of $v_2h_i$ may degenerate to $v_2$.
We place the neighbors in order along this designated part of~$A$, drawing the relevant edges as line segments.
This implies that all these neighbors obtain a line segment that extends to intersect the bottom arc, maintaining the invariant.
We position one neighbor to be a continuation of segment $\ell_{h_i}$, which by the invariant must extend to intersect the designated part of $A$ as well.
\begin{figure}[t]
\centering
\subcaptionbox{\label{fig:determinearc-1}}{\centering\includegraphics[page=1]{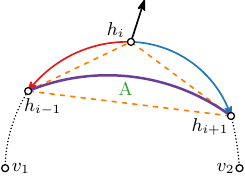}}
\hfil
\subcaptionbox{\label{fig:determinearc-2}}{\centering\includegraphics[page=2]{determinearc}}
\caption{(a) Arc $A$ lies inside the dashed triangle $h_{i-1}h_ih_{i+1}$. (b) Undrawn neighbors of $h_i$ are placed on $A$, in the section determined by $v_1$ and $v_2$. One neighbor is placed to align with $\ell_{h_i}$ towards a predecessor of $h_i$.}
\label{fig:determinearc}
\end{figure}
Two examples of fully drawn graphs are given in Fig.~\ref{fig:example}.
Note that the angular resolution of these drawings decreases very rapidly as~$n$ increases.
Resolving this issue, e.g. through ensuring a polynomial-size grid, remains an open problem.
Another open problem is how many arcs we need if we restrict solutions to a polynomial-size grid.

\begin{theorem}\label{thm:triangle-arc}
	Every triangulation admits a circular arc drawing that uses at most $(5n-11)/3 = 5e/9 - 1/3$ arcs. This drawing can be computed in $O(n)$ time.
\end{theorem}
\begin{proof}
We first analyze the number of circular arcs in our drawing.
We perform our algorithm using the canonical order induced by the canonical ordering tree in the minimal Schnyder realizer having the fewest leaves; without loss of generality, let~$T_n$ be this tree. Recall that $T_n$ has at most $(2n - 5 - \Delta_0)/3$ leaves; since $\Delta_0 \geq 0$, we simplify this to $(2n - 5)/3$ for the remainder of the analysis.
We start with one circle and subsequently process $v_n, \ldots, v_4$, adding one circular arc per vertex (representing edges in $T_1$ and $T_2$) and a number of line segments (representing edges in $T_n$).
Note that processing~$v_3$ has no effect since the edge~$v_1v_2$ is the bottom arc.
Counting the circle as one arc, we thus have $n - 2$ arcs in total.
At every vertex in $T_n$, one incoming edge is collinear with the outgoing one towards the root.
Hence, we charge each line segment uniquely to a leaf of $T_n$: there are at most $(2n-5)/3$ segments.

Thus, the total visual complexity is at most $n - 2  + (2n - 5)/3 = (5n - 11)/3$.
In particular, this shows that, with circular arcs, we obtain greater expressive power for a nontrivial class of graphs in comparison to the $2n$ lower bound that is known for drawing triangulations with line segments.
Since a triangulation has $e = 3n - 6$ edges, the visual complexity can also be expressed as $5e/9 - 1/3$.

We complete the proof by considering the time bound.
The canonical order can be computed in linear time~\cite{dpp-hdpgg-C90}.
We then process the vertices in reverse canonical order. Processing a vertex $v$ of degree $d$ takes $O(1 + d)$ time: we need only a constant number of computations to place arc $A$ and find the subarc where we can place the undrawn neighbors of $v$; the placement of these takes $O(1)$ per neighbor.
Since the graph is planar, the sum over degrees is linear, and thus we obtain an $O(n)$ algorithm for computing this drawing.
\end{proof}

\begin{figure}[t]
\centering
\includegraphics{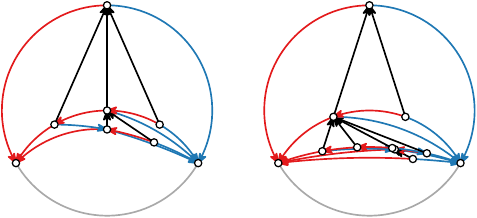}
\caption{Drawings produced by the algorithm with $n = 7$ (left) and $n = 10$ (right).}
\label{fig:example}
\end{figure}

This bound readily improves upon the result for line segments ($7e/9-10/3$) by Durocher and Mondal \cite{durocher2014}.
Schulz~\cite{schulz2015} proved an upper bound of $2e/3$ arcs.
The bound above is an improvement on this result, though only for triangulations.

\paragraph{Degrees of freedom.}
One circular arc has five degrees of freedom (DoF), which is one more than a line segment.
Hence, our algorithm with circular arcs uses at most $5 \cdot (5n - 11)/ 3 = (25n - 55)/3$ DoF, even if we disregard any DoF reduction arising from the need to have arcs coincide at vertices.
This remains an improvement over the result of Durocher and Mondal~\cite{durocher2014}, using at most $4 \cdot (7n - 10)/3 = (28n - 40)/3$ DoF.
The lower bound for line segments ($4 \cdot 2n = 24n/3$) is lower than what we seem to achieve with our algorithm.
However, our algorithm uses line segments rather than arcs to draw the tree~$T_n$.
Thus, the actual DoF employed by the algorithm is $5(n - 2) + 4 \cdot (2n - 5)/3 = (23n - 50)/3$, which is in fact below the lower bound for line segments.

\paragraph{4-connected triangulations.}
We may further follow the rationale of Durocher and Mondal~\cite{durocher2014} by applying a result by Zhang and He~\cite{zhang2005}.
Using regular edge labelings, they proved that a 4-connected triangulation admits a canonical ordering tree with at most $\lceil(n+1)/2\rceil$ leaves~\cite{zhang2005}.
Applying this to our analysis, we find that our algorithm uses at most $n - 2 + \lceil(n+1)/2\rceil = \lceil (3n - 3)/2\rceil \leq 3n/2 - 1$ arcs.

\begin{theorem}\label{thm:4-conn-triangle-arc}
	Every 4-connected triangulation admits a circular arc drawing that uses at most $3n/2 - 1 = e/2 + 2$ arcs. This drawing can be computed in $O(n)$ time.
\end{theorem}

\section{Planar graphs with circular arcs}
\label{sect:grapharcs}

The algorithm for triangulations of the previous section easily adapts to draw a general planar graph~$G$ with $n\ge 3$ vertices and~$e$ edges.
As connected components can be drawn independently, we assume~$G$ is connected.
We need to only triangulate $G$, thereby adding $3n - e - 6$ chords; this takes linear time and can immediately also produce the necessary canonical order (e.g. Section 6 of \cite{kantthesis}).
We then run the algorithm described in Theorem~\ref{thm:triangle-arc}.
Finally, we remove the chords from the drawing.
Each chord may split an arc into two arcs, thereby increasing the total complexity by one.
From Theorem~\ref{thm:triangle-arc}, it follows that we obtain a drawing of $G$ using $(5n/3 - 11/3) + (3n - e - 6) = 14n/3 - e - 29/3$ arcs.
Since $e$ is at least $n-1$, a very rough upper bound is $11n/3 - 26/3$.

\begin{theorem}\label{thm:planar-arc}
	Every planar graph with $n\ge 3$ admits a circular arc drawing with at most $14n/3 - e - 29/3$ arcs. This drawing can be computed in $O(n)$ time.
\end{theorem}

Again, this bound readily improves upon the upper bound for line segments ($16n/3-e-28/3$) by Durocher and Mondal \cite{durocher2014}.
Provided the graph is 3-connected, Schulz's~\cite{schulz2015} bound of $2e/3 - 1$ is lower than our bound, but only for sparse-enough graphs having $e < 14n/5 - 26/5$.
However, there are planar graphs that are not 3-connected with as many as $3n - 7$ edges (one less than a triangulation): there is no sparsity for which planar graphs must be 3-connected and Schulz's bound is lower than our result.
In case the original graph $G$ is 4-connected, extending it to a triangulation by adding edges does not violate this property.
Repeating the above analysis using the improved bound of Theorem~\ref{thm:4-conn-triangle-arc} yields us the following result.

\begin{theorem}\label{thm:4-conn-planar-arc}
	Every 4-connected planar graph admits a circular arc drawing with at most $9n/2 - e - 7$ arcs.  This drawing can be computed in $O(n)$ time.
\end{theorem}

Since any planar graph can be drawn trivially with $e$ arcs (or line segments), the above results are an improvement over the trivial bound, only if $e > 14n/3 - e - 29/3$ for planar graphs or $e > 9n/2 - e - 7$ for 4-connected planar graphs. That is, we need $e > 7n/3 - 29/6$ or $e > 9n/4 - 7/2$, fairly dense graphs.

\paragraph{Heuristic improvement.} 
We investigate here a simple improvement upon the above by choosing a ``good triangulation''.
However, we must finally conclude that for worst-case bounds, this improvement is not noticeable.

For the improvement, we observe that at most two pairs of edges at every vertex share an arc: two edges on the horizon and two edges of~$T_n$.
We thus reduce the necessary geometric primitives by choosing a specific triangulation:
for every face $\langle v_1,\ldots,v_k,v_1\rangle$ with $k>3$, we pick a single vertex, say~$v_1$, and triangulate the
face by connecting~$v_1$ to $v_3,\ldots,v_{k-1}$ with temporary chords.
Note that planarity ensures that there is a vertex for every face
such that we can add these edges without creating multi-edges:
the graph induced by $v_1,\ldots,v_k$ is outerplanar, so it has a degree-2 vertex.
This way, when we remove the temporary chords in the final step of the algorithm,
the number of arcs that are split into two can be reduced.

We may even further save on complexity by selecting the same vertex for multiple adjacent faces.
Bose et al.~\cite{bkl-wcoag-cg03} have shown that one can select in $O(n)$ time
a set~$S$ of at most $\lfloor n/2\rfloor$ vertices such that every face is incident
to at least one vertex of this set.
Note that we would need to only cover faces of size $|f| \geq 4$, as triangular faces do not receive any temporary chords. 
We now triangulate each face to one of its incident vertices in $S$. This means that the temporary chords do not increase the visual complexity, except for up to two chords at every vertex in $S$.

We define the total reduction $R$ as the number of arcs that no longer cause a split of an arc in the analysis of general planar graphs.
The complexity bound thus becomes $14n/3 - e - 29/3 - R$.
By the above rationale, we can express $R$ as \[\sum_{f \in F} (|f| - 3) - 2|S| = \sum_{f \in F} |f| - 3|F| - 2|S| = 2e - 3|F| - 2S,\] where $F$ is the set of faces in the graph $G$.
By Euler's formula, we have $|F| = 2 + e - n$ and we know $S \leq n/2$ by the result of Bose et al.~\cite{bkl-wcoag-cg03}.
We this get that $R \geq 2n - e - 6$.

We can thus guarantee that this heuristic provides a reduction if $e < 2n - 6$, that is, for sparse enough graphs.
Under this assumption, we find that the final bound on the number of arcs is
\[\frac{14n}{3} - e - \frac{29}{3} - R = \frac{14n}{3} - e - \frac{29}{3} - 2n + e + 6 = \frac{8n}{3}-\frac{11}{3}.\]
It is clear that the trivial bound of $e$ performs better than this heuristic for rather sparse graphs, that is, when $e > 8n/3-11/3$.
As we need $e < 2n - 6$ for a the worst-case reduction to be positive, we thus need that $8n/3-11/3 < e < 2n - 6$. This is however not possible for any positive $n$.

We must conclude that this improvement has no effect in the worst-case scenario on the complexity bound for general planar graphs.
This is because the need for large faces for the improvement to be noticeable is
directly opposite the need for a dense graph for our bound to be better than the trivial bound.
However, this method can still be used as a heuristic and does reduce visual complexity, when a graph has multiple faces of size at least~$6$, or when we can find a suitable small set $S$.

Improvements on $S$ to cover the non-triangular faces would readily improve the bounds here as well. For example, we may also take the trivial bound $|S| \leq |F|$. However, going through the same analysis tells us that the reduction must be positive if $e < 5n/8 - 10/8$, that is, for even sparser graphs -- the bound by Bose et al.~\cite{bkl-wcoag-cg03}.

\section{Conclusions}

We investigated the visual complexity of graphs: the number of line segments or circular arcs needed to draw a planar graph.

We provided algorithms to construct segment drawings for various planar graph classes, such that the resulting coordinates are bounded integers;
in other words, we bound the height and width of the integer grid needed for the drawing.
In particular, we have shown how to draw trees of $n$ vertices on an $O(n^2)\times O(n^{1.58})$ grid using at most $3n/4$ segments.
This algorithm can be adapted to achieve the optimal $\theta/2$ segments, where $\theta$ is the number of odd degree vertices, if we increase the grid size to be quasi-polynomial.
Moreover, we described an algorithm to draw maximal outerplanar graphs and planar 3-trees on an $O(n)\times O(n^2)$ grid, using at most $3n/2$ and $(8n - 17)/3$ segments respectively.

Finally, we provided a linear-time algorithm to construct a drawing of an $n$-vertex triangulation using at most $(5n-11)/3$ arcs, though coordinates are not restricted to an integer grid.
This readily gives an algorithm for general planar graphs as well; for both cases, having a 4-connected graph reduces the number of arcs needed.

\paragraph{Open problems.}
There remain a number of interesting avenues for further research.
Can we draw trees using $\theta/2$ segments on a polynomially-sized grid?
For arc drawings, no results on bounded grids are known beyond the result for trees by Schulz \cite{schulz2015}.
In general, stronger lower bounds beyond the general counting arguments (see introduction) are missing;
this prevents us from arguing optimality for many cases, beyond those where the lower bounds can actually always be achieved.


\bibliographystyle{abbrvurl}
\bibliography{abbrv,fewarcs}

\end{document}